\documentclass[conference,a4paper]{IEEEtran}
\IEEEoverridecommandlockouts

\usepackage{amsmath,amssymb,amsthm,bbm,cite,color,comment,graphicx,microtype,subfigure,xcolor}
\usepackage[nolist]{acronym}
\usepackage[USenglish]{babel}
\usepackage[T1]{fontenc}
\usepackage[utf8]{inputenc}
\usepackage{algpseudocode}
\usepackage{float}
\usepackage[ruled,vlined]{algorithm2e}

\usepackage[figurename=Fig., labelsep=period, font=small]{caption}

\usepackage{titlesec}
\let\subparagraph\relax
\titlespacing{\section}{0pt}{6pt plus 2pt minus 1pt}{4pt plus 1pt minus 1pt} 
\titlespacing{\subsection}{0pt}{4pt plus 2pt minus 1pt}{2pt plus 1pt minus 1pt} 
\usepackage{tikz,pgfplots}
\usetikzlibrary{arrows,shapes.geometric}
\pgfplotsset{compat=newest}
\pgfplotsset{
/pgfplots/ybar legend/.style={
/pgfplots/legend image code/.code={%
\draw[##1,/tikz/.cd,bar width=0.1cm,yshift=-0.2em,bar shift=0.5*\pgfplotbarwidth]
plot coordinates {(0.5*\pgfplotbarwidth,0.6em) (2.5*\pgfplotbarwidth,0.4em) (4.5*\pgfplotbarwidth,0.2em)};},}}

\graphicspath{{figures/}}
\usepackage{textcomp}
\def\BibTeX{{\rm B\kern-.05em{\sc i\kern-.025em b}\kern-.08em
    T\kern-.1667em\lower.7ex\hbox{E}\kern-.125emX}}

\newcommand{\h}{\mathbf{h}}

\newcommand{\p}{\mathbf{p}}

\newcommand{\y}{\mathbf{y}}
\newcommand{\z}{\mathbf{z}}

\newcommand{\0}{\mathbf{0}}

\newcommand{\A}{\mathbf{A}}

\newcommand{\U}{\mathbf{U}}

\newcommand{\Deltab}{\mathbf{\Delta}}

\newcommand{\setE}{\mathcal{E}}

\newcommand{\Compl}{\mbox{$\mathbb{C}$}}
\newcommand{\Real}{\mbox{$\mathbb{R}$}}
\newcommand{\rmF}{\mathrm{F}}

\newcommand{\argmax}{\operatornamewithlimits{argmax}}

\newcommand{\Diag}{\mathrm{Diag}}

\newcommand{\herm}{\mathrm{H}}

\newcommand{\rank}{\mathrm{rank}}
\renewcommand{\Re}{\mathrm{Re}}

\newcommand{\tr}{\mathrm{tr}}
\newcommand{\tran}{\mathrm{T}}

\definecolor{oulu_blue}{HTML}{23408F}
\definecolor{oulu_green}{HTML}{39B54A}

\definecolor{red}{rgb}{1,0,0}
\definecolor{red_magenta}{rgb}{1,0,0.5}
\definecolor{magenta}{rgb}{1,0,1}
\definecolor{blue_magenta}{rgb}{0.5,0,1}
\definecolor{blue}{rgb}{0,0,1}
\definecolor{blue_cyan}{rgb}{0,0.5,1}
\definecolor{cyan}{rgb}{0,0.75,0.75} 
\definecolor{green_cyan}{rgb}{0,1,0.5}
\definecolor{green}{rgb}{0,1,0}
\definecolor{green_yellow}{rgb}{0.375,0.75,0} 
\definecolor{yellow}{rgb}{1,1,0}
\definecolor{red_yellow}{rgb}{1,0.5,0}

\begin{acronym}
\acro{2D-MUSIC}{two-dimensional multiple signal classification}
\acro{AWGN}{additive white Gaussian noise}
\acro{BS}{base station}
\acro{CSI}{channel state information}
\acro{DoF}{degree of freedom}
\acro{i.i.d.}{independent and identically distributed}
\acro{LoS}{line-of-sight}
\acro{LS}{least-squares}
\acro{MF}{matched filtering}
\acro{MIMO}{multiple-input multiple-output}
\acro{MUSIC}{multiple signal classification}
\acro{mmWave}{millimeter-wave}
\acro{NLoS}{non-line-of-sight}
\acro{PDF}{probability density function}
\acro{SE}{spectral efficiency}
\acro{SINR}{signal-to-interference-plus-noise ratio}
\acro{sub-THz}{sub-terahertz}
\acro{sum-SE}{sum spectral efficiency}
\acro{ULA}{uniform linear array}
\acro{UL}{uplink}
\acro{DL}{downlink}
\acro{SDR}{semidefinite relaxation}
\acro{ZF}{zero forcing}
\acroplural{LMI}[LMIs]{linear matrix inequalities}
\acro{LMI}{linear matrix inequality}
\acro{XL}{extremely large-scale}
\end{acronym}
\newtheorem{proposition}{Proposition}
\title{Robust Near-Field Beam Focusing \\ Under Imperfect Localization}

\author{
\IEEEauthorblockN{Nima Mozaffarikhosravi, Amirhossein Azarbahram, Prathapasinghe Dharmawansa, and Italo Atzeni}
\IEEEauthorblockA{Centre for Wireless Communications, University of Oulu, Finland \\
E-mail: \{nima.mozaffari, amirhossein.azarbahram, prathapasinghe.kaluwadevage, italo.atzeni\}@oulu.fi
\thanks{This work was supported by the Research Council of Finland (336449 Profi6, 348396 HIGH-6G, and 369116 6G~Flagship).}}
}

\begin{document}

\maketitle
\begin{abstract}
The transition to 6G-and-beyond wireless systems with large-scale antenna arrays and high-frequency deployments significantly extends the near-field region, where channels exhibit a strong dependence on user location. While this enables location-based beam focusing as a low-overhead alternative to conventional channel estimation, its performance is highly sensitive to localization errors. In this paper, we study robust near-field beam focusing under imperfect user localization. We explicitly characterize the impact of localization errors on the line-of-sight-dominated channel by deriving a tractable uncertainty model via a first-order Taylor approximation, which captures the coupled effects of distance and angle in near-field propagation. Building on this model, we formulate a max-min signal-to-interference-plus-noise ratio optimization problem that guarantees performance under worst-case channel realizations induced by bounded localization errors. The resulting problem is reformulated into a feasibility problem using semidefinite relaxation. Numerical results demonstrate that the proposed robust design significantly improves the worst-user rate compared to non-robust beam focusing, particularly under high total transmit power levels and large localization error ranges. 
\end{abstract}

\section{Introduction}

Recent advances in 6G-and-beyond systems, prompted by applications such as virtual reality and autonomous driving, have raised the demand for higher data rates, reliable communications, and balanced performance across users in wireless systems \cite{Nguyen2021_6G, Zhengquan2019_Application}. A promising technology to meet these requirements is massive \ac{MIMO}, which improves performance by providing additional beamforming and spatial gain \cite{Emil2017_mMIMO}. As wireless systems shift toward operating at higher frequency bands, antennas become physically smaller, which facilitates the deployment of \ac{XL}-\ac{MIMO} systems. When large arrays are employed at smaller wavelengths, the radiative near-field region, commonly characterized by the Fraunhofer distance, extends significantly and can reach hundreds of meters. In this regime, the channel explicitly depends on the location of the user through spherical wavefront propagation, making location information a useful means of \ac{CSI} acquisition with lower signaling overhead than conventional pilot-based methods \cite{Nima2025_WCL, Ram2025_loc}.

Near-field beamforming (or beam focusing) inherently leverages user location information to achieve highly directional and spatially selective transmissions. However, to effectively exploit this information, accurate user localization becomes critical. Although modern localization techniques can achieve high accuracy, even small estimation errors may result in significant near-field channel mismatches due to the strong sensitivity to distance and, in particular, angle variations at high frequencies \cite{Dinesh2021_poistion_cuncertatiny}. Therefore, instead of relying on near-perfect localization, it is more practical to design beam-focusing strategies that are inherently robust to location uncertainty. Such robust designs can maintain reliable performance without requiring extremely precise localization or exhaustive refinement of user locations, thereby reducing system complexity and signaling overhead.

Robust beamforming under imperfect \ac{CSI} has been extensively studied in the literature. For instance, \cite{Wang2009_WCRobust} investigates worst-case robust \ac{MIMO} transmission by modeling \ac{CSI} uncertainties within bounded sets and formulating \ac{SINR}-constrained optimization problems, where techniques such as \ac{SDR} enable globally optimal solutions under specific conditions. Moreover, \cite{Jafri2025_CellFreeRobustHBF} extends robust beamforming design to cooperative cell-free mmWave \ac{MIMO} networks, where hybrid beamformers are optimized under imperfect \ac{CSI} to mitigate the performance degradation caused by channel uncertainty, while \cite{Taotao2023_Robust_MaxMin} incorporates user location uncertainty into a max-min fairness formulation, demonstrating how spatial errors directly impact multi-user performance. More recently, robustness has been revisited in the context of near-field and \ac{XL}-\ac{MIMO} systems. Specifically, in \cite{wang2026_RobustNF}, a robust near-field \ac{CSI} estimation framework is proposed by combining model-driven and data-driven approaches, including dual-stage optimization and residual refinement, to improve estimation accuracy under complex propagation conditions. Similarly, \cite{Uchimura2026_DeterStoch} studies robust near-field beams that focus on localization errors by formulating ergodic sum-rate maximization problems and solving them through deterministic and stochastic optimization, highlighting the importance of robustness against location-induced \ac{CSI} inaccuracies.

Despite these advances, several important gaps remain underexplored. Although robustness against channel uncertainty is well established for far-field beamforming, existing works typically rely on channel models where the impact of uncertainty is primarily angular. In contrast, near-field channels depend jointly on distance and angle, making the relationship between localization errors and channel mismatch significantly more intricate. To the best of our knowledge, there is no prior work that directly formulates and solves a robust beam-focusing problem in near-field systems by explicitly starting from user location uncertainty and designing robust beamformers based on the channel for the estimated locations. 

In this work, we consider a \ac{DL} \ac{XL}-\ac{MIMO} system in the near-field regime, where imperfect user localization leads to channel uncertainty at the transmitter. To address this challenge, we explicitly characterize the impact of localization errors on the channel by deriving a tractable worst-case channel model through linearizing the near-field channel with respect to localization errors within a bounded uncertainty set. This formulation captures the coupling between distance and angle in near-field propagation and enables a structured representation of the resulting channel uncertainty. Here, we focus on \ac{LoS}-dominated propagation and an ellipsoidal uncertainty set, as they constitute a fundamental starting point for future extensions to multipath propagation scenarios and more practically accurate localization error models. Based on this formulation, we develop a beam-focusing design using \ac{SDR}, which guarantees performance under the worst-case channel realizations induced by location uncertainty and total transmit power constraints. The resulting optimization problem is reformulated in a convex form that can be efficiently solved. The proposed approach is evaluated through a max-min \ac{SINR} optimization framework, demonstrating its effectiveness in improving worst-user worst-case rate performance under varying transmit power levels and uncertainty set sizes, and highlighting its robustness against localization-induced channel mismatches.

\textit{\textbf{Structure:}} Section~\ref{sec:System_model} introduces the system model and the
problem formulation. The proposed robust beam-focusing design is discussed in Section~\ref{Beamforming Design}, while Section~\ref{Numerical_results} provides the numerical results, and Section~\ref{Conclusion} concludes the paper.

\textit{\textbf{Notations:}} Vectors and matrices are denoted by boldface lower and uppercase letters, respectively. The transpose and conjugate transpose operators are denoted by $(\cdot)^{\mathrm T}$ and $(\cdot)^{\herm}$, respectively. The identity matrix of dimension $n$ and the all-zero matrix of proper dimension are denoted by $\mathbf I_n$ and $\mathbf 0$, respectively, whereas $\Diag(\cdot)$ denotes a diagonal matrix formed from its arguments. $\mathbf A \succ \0$ ($\mathbf A \succeq \0$) denotes that $\mathbf A$ is Hermitian positive (semi)definite, whereas $\mathbf A^{\frac{1}{2}}$ is the positive definite square root of $\A$. The Euclidean vector norm is denoted by $\|\cdot\|_{2}$, the corresponding induced matrix norm by $\|\mathbf A\|_2 = \max_{\|\mathbf v\|_{2} \le 1} \|\mathbf A\mathbf v\|_2$, and $\|\cdot\|_{\rmF}$ denotes the Frobenius norm. The real part and magnitude of a complex scalar are denoted by $\Re\{\cdot\}$ and $|\cdot|$, respectively, while $j =\sqrt{-1}$ denotes the imaginary unit. The trace operator is denoted by $\tr(\cdot)$.

\section{System Model and Problem Formulation}\label{sec:System_model}

In this section, we introduce the considered system model, including the channel and location uncertainty models. Building on these elements, we formulate a robust max-min \ac{SINR} optimization problem that accounts for imperfect \ac{CSI} at the \ac{BS} caused by localization errors.

We consider a \ac{DL} \ac{XL}-\ac{MIMO} system where a \ac{BS} equipped with $M$ antennas simultaneously serves $K$ single-antenna users, as depicted in Fig.~\ref{fig:system_model}. Linear precoding is employed at the \ac{BS}, where $\U =[\mathbf u_{1}, \ldots,\mathbf u_{K}] \in \mathbb C^{M \times K}$ denotes the beamforming matrix and $\mathbf u_k \in \Compl^{M \times 1}$ is the beamforming vector dedicated to user~$k$. The transmit power is constrained as $\|\U\|_{\rmF}^{2} \le P_{\textrm{BS}}$, where $P_{\textrm{BS}}$ is the total transmit power budget. The \ac{DL} channel of user~$k$ is denoted by $\mathbf h_k \in \mathbb C^{M \times 1}$.
Accordingly, the instantaneous \ac{SINR} of user~$k$ is given by
\begin{align}
\gamma_k(\U) =
\frac{|\mathbf h_k^{\herm}\mathbf u_k|^2}
{\sum_{j\neq k}|\mathbf h_k^{\herm}\mathbf u_j|^2 + \sigma_k^2},
\label{eq:true_sinr}
\end{align}
where $\sigma_k^2$ is the \ac{AWGN} power.

We consider a two-dimensional Cartesian setting where the \ac{BS} is located at $(0, 0)$ and the location of user~$k$ is denoted by $\mathbf p_k = [x_k, y_k]^\tran \in \mathbb R^2$.

\begin{figure}[t]
\centering
\pgfdeclareimage{UE}{figures/fig_UE}

\begin{tikzpicture}[>=latex, scale=0.80, transform shape]

\small

\draw[->, thick, red] (-2.5,-1.25) -- (-2.5,-0.5) node[midway, sloped, above, xshift=-2pt] {\footnotesize $x$-axis};
\draw[->, thick, red] (-2.5,-1.25) -- (-1.75,-1.25) node[midway, sloped, below, xshift=-2pt] {\footnotesize $y$-axis};

\draw[very thick] (0,-1.5) -- (0,1.5) node[midway, xshift=-10pt] {BS};
\foreach \y in {-1.5,-1,...,1.5} {
    \node at (0,\y) {$\bullet$};
};

\node[draw, minimum width=1.4cm, minimum height=0.65cm, fill=gray!20, ellipse, rotate=45, transform shape] () at (2,{2*tan(45)}) {};

\node[draw, minimum width=1.45cm, minimum height=0.75cm, fill=gray!20, ellipse, rotate=-20, transform shape] () at (3,{3*tan(-20)}) {};

\node[draw, minimum width=1.75cm, minimum height=0.85cm, fill=gray!20, ellipse, rotate=15, transform shape] () at (4,{4*tan(15)}) {};

\draw[<->, black] (-0.25,1) -- (-0.25,1.5) node[midway, anchor=east] {$d$};
\draw[<->, black]  (-0.7,-1.5) -- (-0.7,1.5) node[midway, anchor=east] {$D$};

\draw[thin,dotted] (0,0) -- (2,{2*tan(45)});
\draw node[scale=0.015] (UE_1) at (2,{2*tan(45)}) {\pgfbox[center,center]{\pgfuseimage{UE}}};

\draw[thin,dotted] (0,0) -- (3,{3*tan(-20)});
\draw node[scale=0.015] (UE_2) at (3,{3*tan(-20)}) {\pgfbox[center,center]{\pgfuseimage{UE}}};

\draw[thin,dotted] (0,0) -- (4,{4*tan(15)});
\draw node[scale=0.015] (UE_3) at (4,{4*tan(15)}) {\pgfbox[center,center]{\pgfuseimage{UE}}};
\node[anchor=west,xshift=1mm] at (4,{4*tan(15)}) {$\hat{\p}_{k}$};
\node[anchor=west,xshift=1mm,yshift=7mm] at (4,{4*tan(15)}) {$\setE_{k}$};

\end{tikzpicture}
\caption{Considered system model with ellipsoidal location uncertainty sets.}
\label{fig:system_model}
\end{figure}
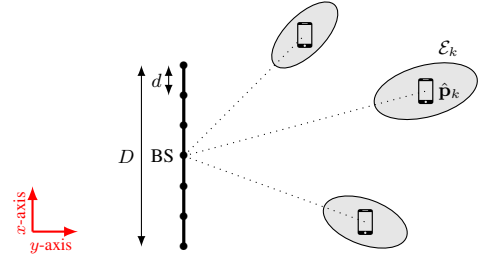

\subsection{Near-Field Channel Model}

We adopt a near-field propagation model where each channel depends explicitly on the user location. As an initial step towards understanding the impact of imperfect user localization, we focus on purely \ac{LoS} propagation, which is representative of many high-frequency scenarios \cite{Italo2025_THz} and allows us to isolate the effect of location uncertainty on the channel; extensions to multipath propagation will be considered in future work. The \ac{BS} is equipped with a \ac{ULA} with antenna spacing $d$ and total length $D = (M-1)d$. Let $f_\textrm{c}$ and $\lambda_\textrm{c} = \frac{c}{f_\textrm{c}}$ denote the carrier frequency and the corresponding wavelength, respectively, where $c$ is the speed of light in vacuum.

We consider a common pathloss factor across all antennas and model the channel between the \ac{BS} and a generic location $\mathbf p =[x,y]^{\tran} \in \Real^{2}$ as \cite{Nima2025_WCL}
\begin{align} \label{channel}
\h^{\textrm{LoS}}(\p)
=
\frac{\lambda_\textrm{c}}{4\pi \|\mathbf p\|_2}
e^{-j\frac{2\pi}{\lambda_\textrm{c}} \|\mathbf p\|_2}
\mathbf b(\mathbf p) \in \mathbb{C}^{M \times 1},
\end{align}
where $\mathbf b(\mathbf p) = [b_0(\mathbf p), \ldots, b_{M-1}(\mathbf p)]^\tran \in \mathbb{C}^{M \times 1}$ denotes the near-field steering vector for the \ac{ULA}. The $m$th element of $\mathbf b(\mathbf p)$ is given by
\begin{align}\label{eq:b_def_cartesian}
b_m(\mathbf p)
=
e^{
-j\frac{2\pi}{\lambda_\textrm{c}} r_m
},
\end{align}
where
\begin{align}
r_m
=
\sqrt{(x - \delta_m d)^2 + y^2}
\label{eq:rn_exact_cartesian}
\end{align}
is the distance between $\mathbf p$ and the $m$th antenna at the \ac{BS}, with $\delta_m = \frac{2m-M+1}{2}$, $\forall m = 0, \ldots, M-1$. Leveraging the second-order Fresnel approximation, which is accurate in the radiating near-field \cite{Cui2024_fresnelchannel}, \eqref{eq:rn_exact_cartesian} can be approximated as
\begin{align}
r_m
\approx
\|\mathbf p\|_2
-
\frac{x}{\|\mathbf p\|_2} \delta_m d
+
\frac{y^2}{2 \|\mathbf p\|_2^3} \delta_m^2 d^2.
\label{eq:fresnel_second_cartesian}
\end{align}
For notational simplicity, the true channel of user~$k$ located at the true location $\mathbf p_k$ and the estimated channel of user~$k$ based on the estimated location $\hat{\mathbf p}_k$ are denoted by $\mathbf h_k = \mathbf h^{\textrm{LoS}}(\mathbf p_k)$ and $\hat{\mathbf h}_k = \mathbf h^{\textrm{LoS}}(\hat{\mathbf p}_k)$, respectively.

\subsection{Location Uncertainty Model}

Due to imperfect user localization, the \ac{BS} has access only to an estimated location
$\hat{\mathbf p}_k \in \Real^{2}$ for each user~$k$, expressed as
\begin{align}\label{Estimated_location}
\hat{\mathbf p}_k =\mathbf p_k + \mathbf e_k \in \mathbb R^2,
\end{align}
where $\mathbf e_k
= [e_{k,x}, e_{k,y}]^\tran
\in \mathbb R^2$ is the localization error vector capturing the uncertainty in the $x$- and $y$-coordinates. The estimated locations may be obtained at the BS through uplink localization or sensing, or reported by the users.

We model imperfect localization using a bounded error model in which the true location of each user lies within a compact region around its estimate~\cite{Taotao2023_Robust_MaxMin}. In particular, we use ellipsoidal uncertainty sets to capture the different localization accuracies in the radial and angular directions, as shown in Fig.~\ref{fig:system_model}. Accordingly, the localization error vector corresponding to user~$k$ belongs to the location uncertainty set
\begin{align}
\mathcal E_k
=
\{
\mathbf e_k \in \mathbb R^2 :
\mathbf e_k^\tran \mathbf Q_k^{-1} \mathbf e_k \le 1
\}, \ \mathbf Q_k \succ \mathbf \0,
\label{eq:ellipsoidal_cartesian}
\end{align}
with size, shape, and orientation determined by $\mathbf Q_k \in \Real^{2 \times 2}$. The bounded ellipsoidal model adopted here provides a tractable representation of the location uncertainty set, making it particularly suitable for robust optimization. The stochastic modeling of localization errors, which may provide a more realistic but less tractable characterization, will be investigated in future work.

\subsection{Problem Formulation} \label{beamfocusing}

Assuming that the \ac{BS} has access only to the estimated user locations in \eqref{Estimated_location}, which results in imperfect \ac{CSI}, our goal is to design the transmit beam-focusing vectors at the \ac{BS} to maximize the worst-case \ac{SINR} among all users. Since the near-field \ac{LoS} channel is determined by the user location through the steering vector, localization errors translate into structured perturbations of the channel. Specifically, substituting \eqref{Estimated_location} into the near-field channel model in \eqref{channel} shows that the induced channel perturbation is generally a nonlinear function of the localization error due to the distance-dependent pathloss and spherical-wave phase terms. Since the location uncertainty enters the near-field channel model in a nonlinear manner, directly treating it in the location domain is analytically challenging. To obtain a more tractable formulation, we instead describe the uncertainty in the channel domain by expressing the estimated channel as \cite{Wang2009_WCRobust}
\begin{align}
\hat{\mathbf h}_k = \mathbf h_k + \mathbf z_k \in \mathbb C^{M \times 1},
\label{eq:channel_uncertainty}
\end{align}
where $\mathbf z_k \in \mathbb C^{M \times 1}$ represents the channel estimation error induced by imperfect user localization. Accordingly, robustness is ensured for each user by considering all the channels that are consistent with the available estimated user location, i.e., all possible channels induced by the user's location uncertainty set. The resulting channel uncertainty set of user~$k$, centered at the estimated channel $\hat{\h}_{k}$, is given by
\begin{align} \label{channel_uncertainty_set}
\mathcal H_k
=
\{
\mathbf h_k \in \Compl^{M \times 1} :
\mathbf h_k = \hat{\mathbf h}_k - \mathbf z_k, \
\mathbf z_k \in \mathcal Z_k
\},
\end{align}
where the set $\mathcal Z_k$ is implicitly induced by the location uncertainty set $\mathcal E_k$ defined in \eqref{eq:ellipsoidal_cartesian} through the nonlinear near-field channel mapping. Despite the convexity of $\mathcal E_k$, the resulting channel uncertainty set in \eqref{channel_uncertainty_set} is nonconvex. Therefore, in Section~\ref{Reformulation}, we derive a tractable approximation that yields an affine uncertainty representation in the channel domain, enabling a convex robust reformulation.

Finally, the considered robust max-min \ac{SINR} optimization problem, which maximizes the worst-user worst-case \ac{SINR}, is formulated as 
\begin{align}
\begin{array}{cl}
\displaystyle \underset{\U}{\mathrm{maximize}} 
& \displaystyle \min_{k = 1,\ldots,K} \; \min_{\mathbf h_k \in \mathcal H_k} \; \gamma_k(\U) \\
\mathrm{s.t.} 
& \displaystyle \|\U\|_{\rmF}^{2} \le P_{\textrm{BS}}.
\end{array}
\label{eq:maxmin_robust}
\end{align}

\section{Robust Beam-focusing Design}\label{Beamforming Design}

In this section, we develop an efficient solution method for the robust max-min \ac{SINR} optimization problem in \eqref{eq:maxmin_robust}. Specifically, we first introduce a tractable channel uncertainty model and then apply \ac{SDR} to the resulting approximate problem.

\subsection{Tractable Channel Uncertainty Model}\label{Reformulation}

Here, we develop a tractable approximation of the estimated channel in \eqref{eq:channel_uncertainty}. The key difficulty is that the channel uncertainty set $\mathcal H_k$ is not given explicitly, but arises implicitly from errors in the estimated location of user~$k$. As a result, the channels and corresponding \ac{SINR} depend nonlinearly on the unknown user locations. To overcome this challenge, we proceed in two steps. First, we establish an explicit approximate relationship between localization errors and channel estimation errors by using a first-order Taylor expansion of \eqref{channel}. Second, since the resulting sensitivity matrix still depends on the underlying user location, we characterize a worst-case channel perturbation over the ellipsoidal location uncertainty set and use it to define a conservative, location-independent uncertainty model for robust beam-focusing design.

We begin by linearizing the channel around the true user location $\mathbf p_k = [x_k,y_k]^\tran$, corresponding to $\mathbf e_k = \mathbf 0$. Based on \eqref{channel}, a first-order Taylor expansion of the channel with respect to the generic location $\mathbf p = [x,y]^\tran$ yields
\begin{align}\label{taylor_expansion}
\h^{\textrm{LoS}}(\p)
\approx
\mathbf h_k
+
\big.
\frac{\partial \h^{\textrm{LoS}}(\p)}{\partial x}
(x - x_k)
+
\big.
\frac{\partial \h^{\textrm{LoS}}(\p)}{\partial y}
(y - y_k). 
\end{align}
Defining $r = \|\mathbf p\|_2 = \sqrt{x^2 + y^2}$, the partial derivative of the channel with respect to $x$ is given by
\begin{align}\label{Partial_x}
\mathbf g_{x}(\p)
& = \big.
\frac{\partial \h^{\textrm{LoS}}(\p)}{\partial x} \\
\nonumber & =
\frac{\lambda_\textrm{c}}{4\pi r}
e^{-j \frac{2\pi}{\lambda_\textrm{c}} r}
\bigg(
\bigg(-\frac{1}{r}
- j \frac{2\pi}{\lambda_\textrm{c}}\bigg)
\frac{x}{r}\mathbf I_M
\\
& \phantom{=} \ +
j \frac{2\pi}{\lambda_\textrm{c}}
\bigg(
d\frac{y^2}{r^3}
\mathbf \Deltab
+
d^2\frac{3 x y^2}{2 r^5}
\mathbf \Deltab^2
\bigg)
\bigg)
\mathbf b(\mathbf p),
\label{eq:dh_dx}
\end{align}
with
$\mathbf \Deltab
=
\mathrm{Diag}(\delta_0,\ldots,\delta_{M-1})  \in \mathbb R^{M \times M}$. Similarly, the partial derivative with respect to $y$ is given by
\begin{align}\label{Partial_y}
\mathbf g_{y}(\p)
& = \big.
\frac{\partial \h^{\textrm{LoS}}(\p)}{\partial y} \\
\nonumber & =
\frac{\lambda_\textrm{c}}{4\pi r}
e^{-j \frac{2\pi}{\lambda_\textrm{c}} r}
\bigg(
\bigg(-\frac{1}{r}
- j \frac{2\pi}{\lambda_\textrm{c}}\bigg)
\frac{y}{r}\mathbf I_M
\\
& \phantom{=} \ -
j \frac{2\pi}{\lambda_\textrm{c}}
\bigg(
d\frac{x y}{r^3}
\mathbf \Deltab
+
d^2\bigg(
\frac{y}{r^3}
-
\frac{3 y^3}{2 r^5}
\bigg)
\mathbf \Deltab^2
\bigg)
\bigg)
\mathbf b(\mathbf p).
\label{eq:dh_dy}
\end{align}
Then, \eqref{eq:dh_dx} and \eqref{eq:dh_dy} are used to construct the sensitivity matrix
\begin{align}\label{sensetivity_matrix}
\mathbf G_k(\mathbf p)
=
\big[\mathbf g_{x}(\p), \mathbf g_{y}(\p)\big]
\in \mathbb{C}^{M\times 2}.
\end{align}
Recalling the relation $\hat{\mathbf{p}}_k  - \mathbf p_k  = \mathbf e_{k}$, we substitute $\hat{\mathbf p}_k$ into \eqref{sensetivity_matrix} to approximate the estimated channel in \eqref{eq:channel_uncertainty} as
\begin{align}
\hat{\mathbf h}_k
\approx
\mathbf h_k
+
\underbrace{\mathbf G_k(\hat{\mathbf p}_k) \mathbf e_k}_{\approx \z_{k}},
\label{eq:first_order_matrix_cart}
\end{align}
which establishes an approximate affine mapping from the localization error $\mathbf e_k$ to the channel perturbation $\z_{k}$. The neglected Taylor remainder is of order $\mathcal O(\|\mathbf e_k\|_2^2)$, making the approximation accurate for small bounded localization errors.

Although \eqref{eq:first_order_matrix_cart} provides an affine perturbation model, the sensitivity matrix still depends on the estimated user location $\hat{\mathbf p}_k$, which yields an approximate perturbation model that depends on the true user location $\p_k$. This dependence prevents a direct characterization of the channel uncertainty set \eqref{channel_uncertainty_set} in a form that is independent of the underlying user location, which is required to express the objective in \eqref{eq:maxmin_robust} solely in terms of the channels and enable a tractable worst-case reformulation. To address this issue, we next characterize the worst-case channel perturbation induced by the ellipsoidal location uncertainty set in \eqref{eq:ellipsoidal_cartesian} and use it to identify a conservative evaluation point for the sensitivity matrix.

\begin{proposition}\label{prop:maxEstimationError}
Under the first-order approximation in \eqref{eq:first_order_matrix_cart}, the worst-case magnitude of the channel estimation error of user~$k$ over the ellipsoidal location uncertainty set in \eqref{eq:ellipsoidal_cartesian} is given by\footnote{Note that the left- and right-hand sides of \eqref{eq:maxEstimationError} are expressed in terms of the Euclidean vector norm and the corresponding induced matrix norm, respectively.}
\begin{align}
\max_{\mathbf e_k \in \mathcal E_k}
\big\|
\mathbf G_k(\mathbf p)\mathbf e_k
\big\|_2
=
\big\|
\mathbf G_k(\mathbf p)\mathbf Q_k^{\frac{1}{2}}
\big\|_2.
\label{eq:maxEstimationError}
\end{align}
\end{proposition}

\begin{proof}
Given $\mathbf Q_k \succ \mathbf 0$, we express the localization error vector as $\mathbf e_k = \mathbf Q_k^{\frac{1}{2}}\mathbf v_k$ such that
\begin{align}
\mathbf e_k^{\tran}\mathbf Q_k^{-1}\mathbf e_k
=
\mathbf v_k^{\tran}\mathbf v_k
=
\|\mathbf v_k\|_2^2.
\end{align}
Hence, the constraint
$\mathbf e_k^{\tran}\mathbf Q_k^{-1}\mathbf e_k \le 1$
is equivalent to
$\|\mathbf v_k\|_2 \le 1$, which yields
\begin{align}
\max_{\mathbf e_k \in \mathcal E_k}
\|\mathbf G_k(\mathbf p)\mathbf e_k\|_2
&=
\max_{\|\mathbf v_k\|_2 \le 1}
\|\mathbf G_k(\mathbf p)\mathbf Q_k^{\frac{1}{2}}\mathbf v_k\|_2.
\end{align}
Finally, \eqref{eq:maxEstimationError} follows directly from the definition of the induced matrix $2$-norm.
\end{proof}
Building on Proposition~\ref{prop:maxEstimationError}, we define a conservative evaluation point for the sensitivity matrix by selecting the location within the location uncertainty set that maximizes the magnitude of the channel perturbation in \eqref{eq:first_order_matrix_cart}. Such worst-case location for user~$k$ is defined as
\begin{align}
\mathbf p_k^{\textrm{wc}}
=
\underset{\p \, : \, (\hat{\mathbf p}_k - \mathbf p) \in \mathcal E_k}{\argmax}
\big\|
\mathbf G_k(\mathbf p)\mathbf Q_k^{\frac{1}{2}}
\big\|_2,
\end{align}
and evaluating the sensitivity matrix at $\mathbf p_k^{\textrm{wc}}$ yields the location-independent, worst-case sensitivity matrix $\mathbf G_k^{\textrm{wc}} = \mathbf G_k(\mathbf p_k^{\textrm{wc}})$. Consequently, the channel estimation error in \eqref{eq:first_order_matrix_cart} is replaced by $\mathbf G_k^{\textrm{wc}} \mathbf e_k$, resulting in a worst-case additive uncertainty model that depends only on the bounded error vector $\mathbf e_k$. Then, the corresponding worst-case channel uncertainty set is defined as
\begin{align}
\mathcal H_k^{\textrm{wc}}
=
\{
\mathbf h_k \in \Compl^{M \times 1}
:
\mathbf h_k
=
\hat{\mathbf h}_k
-
\mathbf G_k^{\textrm{wc}}\mathbf e_k,
\
\mathbf e_k^{\tran}\mathbf Q_k^{-1}\mathbf e_k \le 1
\}.
\label{eq:worst_case_channel_uncertainty_set}
\end{align}
In the following, we replace the channel uncertainty set in \eqref{channel_uncertainty_set} with its worst-case counterpart in \eqref{eq:worst_case_channel_uncertainty_set}.

\subsection{SDR-Based Reformulation}\label{SDR-Based Robust Beam focusing}

Consider the robust max-min \ac{SINR} optimization problem in \eqref{eq:maxmin_robust} with $\mathcal H_k$ replaced by $\mathcal H_k^{\textrm{wc}}$ in \eqref{eq:worst_case_channel_uncertainty_set}, $\forall k = 1, \ldots, K$. This can be rewritten in epigraph form as
\begin{subequations}\label{eq:epigraph_robust}
\begin{align}
\hspace{-2mm} \underset{\{\U\}, \gamma}{\mathrm{maximize}} 
& \ \ \gamma \label{eq:epigraph_robust_a}\\
\mathrm{s.t.} \ \ \ \
& \ \ \gamma_k(\U) \ge \gamma, \ \forall \mathbf h_k \in \mathcal H_k^{\textrm{wc}},\ \forall k = 1, \ldots, K,
\label{eq:epigraph_robust_b}\\
& \ \ \displaystyle \|\U\|_{\rmF}^{2} \le P_{\textrm{BS}}.
\label{eq:epigraph_robust_c}
\end{align}
\end{subequations}
This problem is quasi-convex in $\gamma$, since for any fixed $\gamma$ the \ac{SINR} constraints in \eqref{eq:epigraph_robust_b} define a convex feasibility set. Therefore, the optimal value of $\gamma$ can be efficiently obtained via a bisection search, where at each iteration a convex feasibility problem is solved \cite[Ch.~4.2.5]{boyd2004_convex}.

Let us introduce $\mathbf W_k = \mathbf u_k \mathbf u_k^{\herm} \in \mathbb{C}^{M\times M}$ as the rank-one beam-focusing matrix of user~$k$ and rewrite the total transmit power constraint in \eqref{eq:epigraph_robust_c} as
\begin{align}
\sum_{k=1}^{K} \tr(\mathbf W_k)
\le
P_{\textrm{BS}}, \ \mathbf W_k \succeq \0.
\end{align}
Furthermore, for a given $\gamma \ge 0$, let us define $\mathbf S_k(\gamma)
=
\mathbf W_k
-
\gamma \sum_{j \neq k} \mathbf W_j \in \mathbb{C}^{M\times M}$ and express the \ac{SINR} constraints in \eqref{eq:epigraph_robust_b} as
\begin{align}
\mathbf h_k^{\herm}
\mathbf S_k(\gamma)
\mathbf h_k
\ge
\gamma \sigma_k^2, \ \forall \mathbf h_k \in \mathcal H_k^{\textrm{wc}},\ \forall k = 1, \ldots, K.
\label{eq:robust_quadratic}
\end{align}
Using \eqref{eq:worst_case_channel_uncertainty_set} and substituting the corresponding channels into \eqref{eq:robust_quadratic} allows explicit characterization of the impact of imperfect localization. Since $\mathbf S_k(\gamma)$ is Hermitian, expanding
$
(\hat{\mathbf h}_k-\mathbf G_k^{\textrm{wc}}\mathbf e_k)^{\herm}
\mathbf S_k(\gamma)
(\hat{\mathbf h}_k-\mathbf G_k^{\textrm{wc}}\mathbf e_k)
$
and collecting the terms with respect to the localization error vector
$\mathbf e_k$ yields a quadratic function of $\mathbf e_k$.
Consequently, \eqref{eq:robust_quadratic} can be rewritten as
\begin{align}
\mathbf e_k^{\tran} \mathbf A_k(\gamma) \mathbf e_k
+
2 \Re \big\{ \boldsymbol{\kappa}_k(\gamma)^{\herm} & \mathbf e_k \big\}
+
\chi_k(\gamma)
\ge 0, \nonumber \\ & \forall \mathbf e_k \in \mathcal E_k, \ \forall k = 1, \ldots, K,
\label{eq:quad_expanded}
\end{align}
where we have defined
\begin{align}
\mathbf A_k(\gamma)
&=
(\mathbf G_k^{\textrm{wc}})^{\herm}
\mathbf S_k(\gamma)
\mathbf G_k^{\textrm{wc}}
\in \mathbb C^{2\times 2}, \\
\boldsymbol{\kappa}_k(\gamma)
&=
-
(\mathbf G_k^{\textrm{wc}})^{\herm}
\mathbf S_k(\gamma)
\hat{\mathbf h}_k
\in \mathbb C^{2 \times 1}, \\
\chi_k(\gamma)
&=
\hat{\mathbf h}_k^{\herm}
\mathbf S_k(\gamma)
\hat{\mathbf h}_k
-
\gamma \sigma_k^2
\in \mathbb R .
\end{align}

\begin{algorithm}[t]
\small
\caption{Proposed robust beam-focusing design via \ac{SDR} and bisection}
\label{alg:robust_beamforming}
\LinesNumbered

\KwIn{Estimated user locations $\{\hat{\mathbf p}_k\}_{k=1}^K$,
uncertainty matrices $\{\mathbf Q_k\}_{k=1}^K$,
noise powers $\{\sigma_k^2\}_{k=1}^K$,
$P_{\textrm{BS}}$, bisection tolerance $\epsilon$,
initial SINR value $\gamma_0$, growth factor $\rho>1$,
and upper limit $\gamma_{\textrm{max}}$.}

\KwOut{Beam-focusing vectors $\{\mathbf u_k\}_{k=1}^K$.}

\For{$k=1,\ldots,K$}{
Compute estimated channel $\hat{\mathbf h}_k$ using \eqref{channel}\;
Compute sensitivity matrix $\mathbf G_k(\mathbf p)$ using
\eqref{eq:dh_dx}--\eqref{eq:dh_dy}\;
Determine worst-case sensitivity $\mathbf G_k^{\textrm{wc}}$\;
}

\tcp{Determine the initail \ac{SINR} bracket}
Set $\gamma_{\textrm{min}}\gets 0$ and $\gamma\gets\gamma_0$\;

\While{$\gamma\leq\gamma_{\textnormal{max}}$}{
Solve feasibility problem~\eqref{eq:feasibility_problem} for $\gamma$\;

\eIf{problem~\eqref{eq:feasibility_problem} is feasible}{
    Set $\gamma_{\textrm{min}}\gets\gamma$ and store
    $\{\mathbf W_k^\star\}$\;
    Set $\gamma\gets\min\{\rho\gamma,\gamma_{\textrm{max}}\}$\;
}{
    Set $\gamma_{\textrm{max}}\gets\gamma$\;
    \textbf{break}\;
}
}

\tcp{Refine the bracket by bisection}
\While{$\gamma_{\textnormal{max}}-\gamma_{\textnormal{min}}>\epsilon$}{
Set $\gamma\gets
(\gamma_{\textrm{min}}+\gamma_{\textrm{max}})/2$\;

Solve feasibility problem~\eqref{eq:feasibility_problem}\;

\eIf{problem~\eqref{eq:feasibility_problem} is feasible}{
    Set $\gamma_{\textrm{min}}\gets\gamma$ and store
    $\{\mathbf W_k^\star\}$\;
}{
    Set $\gamma_{\textrm{max}}\gets\gamma$\;
}
}

\For{$k=1,\ldots,K$}{
Recover beam-focusing vectors from $\mathbf W_k^\star$
via Gaussian randomization\;
}

\Return{$\{\mathbf u_k\}_{k=1}^K$}\;
\end{algorithm}

Now, since $\mathbf e_k$ lies in an ellipsoidal set, \eqref{eq:quad_expanded} can be converted into a set of finite-dimensional linear matrix inequalities via the S-procedure~\cite{song2012_robust}.
Specifically, a sufficient condition is the existence of $\{\nu_k \ge 0\}_{k=1}^{K}$ satisfying
\begin{align}
\begin{bmatrix}
\mathbf A_k(\gamma) + \nu_k \mathbf Q_k^{-1}
&
\boldsymbol{\kappa}_k(\gamma)
\\
\boldsymbol{\kappa}_k(\gamma)^{\herm}
&
\chi_k(\gamma) - \nu_k
\end{bmatrix}
\succeq \0, \ \forall k = 1, \ldots, K.
\label{eq:lmi_constraint}
\end{align}
As \eqref{eq:lmi_constraint} is parameterized by the target \ac{SINR} level $\gamma$, for a fixed $\gamma$, the max-min \ac{SINR} maximization problem can be solved through a feasibility check over the rank-one beam-focusing matrices. Specifically, for a fixed $\gamma$, the feasibility problem is formulated as
\begin{align}
\begin{array}{cl}
\mathrm{find}
& \{\mathbf W_k \succeq \0\}_{k=1}^{K},\ \{\nu_k \ge 0\}_{k=1}^{K} \\[3pt]
\mathrm{s.t.}
&
\begin{bmatrix}
\mathbf A_k(\gamma) + \nu_k \mathbf Q_k^{-1}
&
\! \! \! \! \boldsymbol{\kappa}_k(\gamma)
\\
\boldsymbol{\kappa}_k(\gamma)^{\herm}
&
\! \! \! \! \chi_k(\gamma) - \nu_k
\end{bmatrix}
\succeq \0, \ \forall k = 1, \ldots, K,
\\[9pt]
&
\displaystyle
\sum_{k=1}^{K}\tr(\mathbf W_k)
\le P_{\textrm{BS}},
\\[8pt]
&
\rank(\mathbf W_k)=1, \ \forall k = 1, \ldots, K,
\end{array}
\label{eq:feasibility_problem}
\end{align}
where the rank-one constraints make the problem nonconvex.
By applying \ac{SDR}, we drop these rank-one constraints, and the resulting problem becomes a semidefinite program that can be solved efficiently by first determining a feasible--infeasible interval for $\gamma$ and then applying bisection within this interval to determine the optimal target \ac{SINR} $\gamma$. Since the relaxed feasibility problem is convex for any fixed $\gamma$, bisection converges to the globally optimal value of the relaxed quasiconvex problem up to the chosen tolerance $\epsilon$~\cite[Ch.~4.2.5]{boyd2004_convex}. After solving the relaxed problem, the beam-focusing vectors $\{\mathbf u_k\}_{k=1}^K$ are recovered from $\{\mathbf W_k\}_{k=1}^K$ using Gaussian randomization. In our simulations, however, the obtained matrices are observed to be nearly rank-one, and the beam-focusing vectors can therefore be directly extracted from the principal eigenvectors (scaled by the square root of their corresponding eigenvalues) with negligible performance loss. The proposed robust beam-focusing design is summarized in Algorithm~\ref{alg:robust_beamforming}.

\section{Numerical Results}\label{Numerical_results}

In this section, we evaluate the performance of the proposed robust beam-focusing design, examining the impact of both the transmit power budget $P_{\textrm{BS}}$ and the level of location uncertainty. Since the proposed optimization maximizes the worst-user worst-case \ac{SINR}, the reported performance is expressed in terms of the corresponding worst-user worst-case rate obtained from the optimized \ac{SINR}. We adopt as the primary baseline the non-robust counterpart of \eqref{eq:feasibility_problem}, obtained by treating the estimated channels $\hat{\mathbf h}_k$ as perfect, since existing robust beam-focusing methods such as \cite{Uchimura2026_DeterStoch} employ different objectives and uncertainty models and therefore do not permit a direct comparison under the same system model. We consider a \ac{BS} equipped with $M=128$ antennas spaced by $d=\frac{\lambda_{\textrm{c}}}{2}$, operating at carrier frequency $f_{\textrm{c}}=16$~GHz (within the upper mid-band) and serving $K=3$ users. The users are located in the near-field region, with radial distances uniformly distributed in $r \in \left[\frac{R_{\textrm{F}}}{6}, \frac{R_{\textrm{F}}}{1.5}\right]$ and angular locations in $\theta \in [-\frac{\pi}{2}, \frac{\pi}{2}]$, where $R_{\textrm{F}}=\frac{2D^{2}}{\lambda_{\textrm{c}}}$ denotes the Fraunhofer distance. A minimum separation constraint $\|\mathbf p_i - \mathbf p_j\|_2 \geq \frac{R_{\textrm{F}}}{6}$ is enforced between any pair of users $i \ne j$. To vary the transmit power using a single normalized parameter, we set the \ac{AWGN} power to $\sigma^2 = 1$ and normalize the channel vectors and their sensitivities with respect to the root-mean-square channel norm across users. The value of $P_{\textrm{BS}}$ is then used in the optimization through the total transmit power constraint, and the resulting beam-focusing vectors implicitly determine how the available power is distributed among the users. The location uncertainty set of each user is modeled as an oriented ellipsoidal set. Specifically, the uncertainty covariance matrix is defined as
\begin{align}
\mathbf{Q}_k = \mathbf{R}_k 
\begin{bmatrix}
\xi_{\textrm{major}}^2 & 0 \\
0 & \xi_{\textrm{minor}}^2
\end{bmatrix}
\mathbf{R}_k^\tran, \ \forall k = 1,\ldots,K,
\end{align}
where $\mathbf{R}_k$ is a rotation matrix whose first column is aligned with the radial direction from the \ac{BS} to user $k$ and the second column is tangential. The major-axis of the uncertainty set $\xi_{\textrm{major}}$ is specified directly in meters, while the minor-axis is set as $\xi_{\textrm{minor}} = \beta \xi_{\textrm{major}}$ with $\beta=0.5$. Moreover, the localization error $\mathbf e_k$ is drawn uniformly from the corresponding ellipsoidal location uncertainty set, and we plot the worst-user worst-case rate using Monte Carlo simulations with $50$ user locations.

Fig.~\ref{fig:power_sweep_worst_user_rate} shows the worst-user worst-case rate versus the transmit power budget. As expected, increasing $P_{\textrm{BS}}$ improves the achieved rates, as the system transitions from a noise-limited to a progressively more interference-limited regime. The robust design consistently outperforms the non-robust baseline across the entire power range, highlighting its ability to account for location uncertainty in the beam-focusing design. The performance gap becomes more pronounced at high transmit power, where the impact of mismatch between the estimated and true user locations leads to increased residual interference in the non-robust scheme.

Fig.~\ref{fig:alpha_sweep_worst_user_rate} plots the worst-user worst-case rate versus the uncertainty level controlled by $\xi_{\textrm{major}}$, with $P_{\textrm{BS}}=35$~dBW. As expected, increasing $\xi_{\textrm{major}}$ degrades the achieved rates due to the growing mismatch between the estimated and true user locations. The proposed robust design also consistently outperforms the non-robust baseline, highlighting its ability to account for location uncertainty in the beam focusing. Moreover, the performance gap widens as $\xi_{\textrm{major}}$ increases, since the non-robust scheme becomes increasingly sensitive to channel mismatch, whereas the robust design maintains better reliability under more challenging localization conditions.

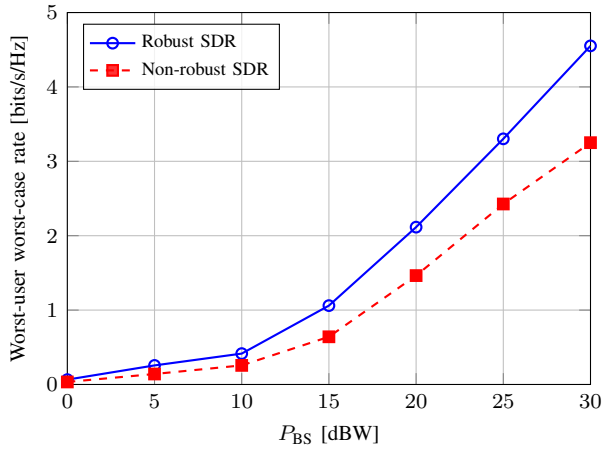
\begin{figure}[t]
    \centering
    \begin{tikzpicture}
    \begin{axis}[
        width=8.5cm, height=6.5cm,
        xmin=0, xmax=30,
        xtick={0, 5, 10, 15, 20, 25, 30},
        ymin=0, ymax=5,
        ytick={0, 1, 2, 3, 4, 5},
        grid=major,
        xlabel={$P_{\textrm{BS}}$ [dBW]},
        ylabel={Worst-user worst-case rate [bits/s/Hz]},
        legend style={
            at={(0.03,0.97)},
            anchor=north west,
            font=\scriptsize,
            inner sep=2pt,
            fill opacity=0.8,
            draw opacity=1,
            text opacity=1
        },
        legend cell align=left,
        ticklabel style={font=\footnotesize},
        x label style={font=\footnotesize},
        y label style={font=\footnotesize},
    ]

    \addplot[blue, thick, mark=o]
    table[x index=0, y index=1, col sep=space] {figures/SNR_sweep_averaged.dat};
    \addlegendentry{Robust \ac{SDR}}

    \addplot[red, thick, dashed, mark=square*, mark options={solid}]
    table[x index=0, y index=2, col sep=space] {figures/SNR_sweep_averaged.dat};
    \addlegendentry{Non-robust \ac{SDR}}

    \end{axis}
    \end{tikzpicture}
    \caption{Worst-user worst-case rate versus the transmit power budget $P_{\textrm{BS}}$ with $\xi_{\textrm{major}} = 1$~m.}
    \label{fig:power_sweep_worst_user_rate} \vspace{-2mm}
\end{figure}

\begin{figure}[t]
    \centering
    \begin{tikzpicture}
    \begin{axis}[
        width=8.5cm, height=6.5cm,
        xmin=0, xmax=1,
        xtick={0,0.2,0.4,0.6, 0.8, 1},
        ymin=3.5, ymax=10,
        ytick={3,4,5,6,7,8, 9, 10},
        grid=major,
        xlabel={$\xi_{\textrm{major}}$ [m]},
        ylabel={Worst-user worst-case rate [bits/s/Hz]},
        legend style={
            at={(0.97,0.97)},
            anchor=north east,
            font=\scriptsize,
            inner sep=2pt,
            fill opacity=0.8,
            draw opacity=1,
            text opacity=1
        },
        legend cell align=left,
        ticklabel style={font=\footnotesize},
        x label style={font=\footnotesize},
        y label style={font=\footnotesize},
    ]

    \addplot[blue, thick, mark=o]
    table[x index=0, y index=1, col sep=space] {figures/major_sweep_averaged.dat};
    \addlegendentry{Robust \ac{SDR}}

    \addplot[red, thick, dashed, mark=square*, mark options={solid}]
    table[x index=0, y index=2, col sep=space] {figures/major_sweep_averaged.dat};
    \addlegendentry{Non-robust \ac{SDR}}

    \end{axis}
    \end{tikzpicture}
    \caption{Worst-user worst-case rate versus location uncertainty level with $P_{\textrm{BS}}=35$~dBW.}
    \label{fig:alpha_sweep_worst_user_rate} \vspace{-2mm}
\end{figure}
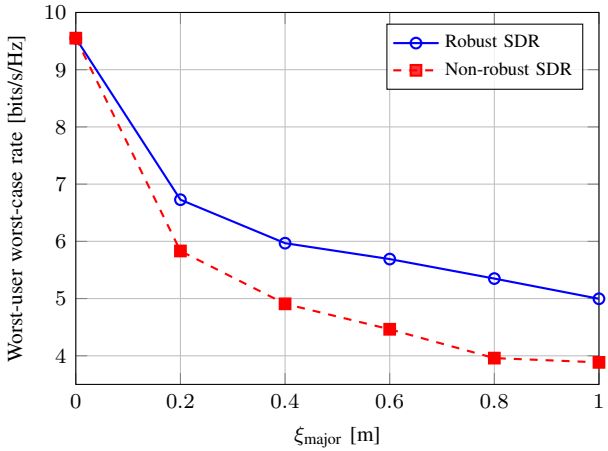

\section{Conclusion}\label{Conclusion}

In this work, we considered a \ac{DL} \ac{XL}-\ac{MIMO} system in the near-field regime with user location uncertainty. In such systems, even small localization errors can significantly degrade performance due to the strong dependence of near-field channels on both distance and angle. Motivated by this, we developed a robust max-min \ac{SINR} beam-focusing framework by linking localization errors to channel uncertainty via a first-order approximation, enabling a tractable worst-case formulation based on \ac{SDR}. Numerical results showed that the proposed design consistently improves the worst-user rate over non-robust methods under different total transmit power levels and localization error ranges. Future work will extend this framework by considering multipath propagation and more realistic location uncertainty models that better capture practical system conditions, while also developing lower-complexity solution methods to reduce the computational burden associated with the semidefinite programming feasibility problem.

\addcontentsline{toc}{chapter}{References}
\bibliographystyle{IEEEtran}
\bibliography{refs_abbr,refs}

\end{document}